\documentclass{amsart}
\usepackage{a4}
\usepackage{graphicx}
\usepackage{amsmath,amssymb,amsthm, color, fancyvrb, array}
\usepackage{ mathrsfs }

\numberwithin{equation}{section}

\theoremstyle{plain}
\newtheorem{thm}{Theorem}[section]
\newtheorem{lem}[thm]{Lemma}
\newtheorem{cor}[thm]{Corollary}
\newtheorem{prop}[thm]{Proposition}
\theoremstyle{definition}
\newtheorem{defn}[thm]{Definition}

\newcommand{\R}{{\mathbb R}}

\newcommand{\di}{\mathrm{div}}
\newcommand{\Ric}{\mathrm{Ric}}
\newcommand{\mr}{\mathring}

\newcommand{\Lich}{\mathrm{Lich}}

\newcommand{\supp}{\mathrm{supp}}

\newcounter{mnotecount}[section]
\let\oldmarginpar\marginpar
\renewcommand\marginpar[1]{\-\oldmarginpar[\raggedleft\footnotesize #1]%
{\raggedright\footnotesize #1}}

\title[A Limit Equation for AC Initial Data]{A Limit Equation Criterion for Applying  the Conformal Method to Asymptotically Cylindrical Initial Data Sets}

\author[J. Dilts]{James Dilts$^1$}
\address{$^1$University of Oregon}
\email{jdilts@uoregon.edu}

\author[J. Leach]{Jeremy Leach$^2$}
\address{$^2$Stanford University}
\email{jleach@math.stanford.edu}

\date{\today}

\begin{document}
\maketitle
\frenchspacing
\begin{abstract} We prove that in a certain class of conformal data on an asymptotically cylindrical manifold, if the conformally decomposed Einstein constraint equations do not admit a solution, then one can always find a nontrivial solution to the limit equation first explored by Dahl, Gicquaud, and Humbert in \cite{DGH11}. We also give an example of a Ricci curvature condition on the manifold which precludes the existence of a solution to this limit equation, showing that such a limit criterion can be a useful tool for studying the Einstein constraint equations on manifolds with asymptotically cylindrical ends.
\end{abstract}

\section{Introduction}

It has been known for over sixty years that Einstein's equations \textit{in vacuo} can be decomposed into a well-posed initial value problem in which a Riemannian $n$-manifold $(M,\overline{g})$ along with a symmetric 2-tensor $\overline{K}$ can be isometrically embedded as a spacelike slice of a spacetime with extrinsic curvature tensor $\overline{K}$ if and only if the Einstein constraint equations
\begin{equation}\label{constrainteq1}
R_{\overline{g}} - |\overline{K}|^2 _{\overline{g}} + (\mathrm{tr}_{\overline{g}} \overline{K})^2 = 0
\end{equation} \begin{equation}\label{constrainteq2}
\di_{\overline{g}} \overline{K} - \nabla \mathrm{tr}_{\overline{g}} \overline{K} = 0
\end{equation}

\noindent are satisfied.

The study of these equations has been an active field of research in recent years, and several approaches have been taken to solve (\ref{constrainteq1})-(\ref{constrainteq2}). By far the most common is the conformal method, in which one specifies a background metric $g$ and looks for a solution of the constraints in the conformal class of $g$. In particular, we seek a scalar function $\phi$ and a vector field $W$ such that $(\overline{g},\overline{K})$ solves (\ref{constrainteq1})-(\ref{constrainteq2}) where this pair is of the form
\begin{equation}\label{conformalg}
\overline{g}_{ij} = \phi^{N-2} g_{ij}
\end{equation} \begin{equation}\label{conformalK}
\overline{K}_{ij} = \frac{\tau}{n} \overline{g}_{ij} + \phi^{-2} (\sigma_{ij} + (LW)_{ij}).
\end{equation}

\noindent Here $\tau = \mathrm{tr}_{\overline{g}} \overline{K}$, $\sigma$ is a specified transverse traceless tensor, $L$ is the conformal Killing operator which we define below, and $N$ is a dimensional constant we will frequently use throughout this paper, given by
\begin{equation*}
N = \frac{2n}{n-2}.
\end{equation*}

It is well known (see, for example, \cite{G12}) that the constraint equations are thus reduced to solving a particular semilinear elliptic system in $\phi$ and $W$, which is often called the LCBY equations in honor of Lichnerowicz, Choquet-Bruhat, and York who first studied them:

\begin{equation}\label{origLich}
  \Delta_g \phi - c_n R_g \phi = b_n \tau^2 \phi^{N-1} - c_n |\sigma + LW|^2 _g\phi^{-N-1}
\end{equation}\begin{equation}\label{origVect}
  \di L W = \frac{n-1}{n} \phi^{N} d\tau.
\end{equation}

\noindent Equation \eqref{origLich} is known as the \textit{Lichnerowicz equation}, and in it we have introduced the dimensional constants
\begin{equation*}
c_n = \frac{n-2}{4(n-1)}, \quad b_n = \frac{n-2}{4n}.
\end{equation*}

\noindent .

Indeed much of the efforts aimed at finding solutions to the constraint equations have been focused on determining the solvability of the system (\ref{origLich})-(\ref{origVect}). One of the earliest systematic studies of these equations was the 1995 paper of J. Isenberg \cite{Isenberg95} which considered the case of constant mean curvature (CMC) $\tau$ on closed manifolds, in which case the LCBY system decouples since the vector equation becomes trivial. In the years since, many authors have sought to generalize Isenberg's results on closed manifolds with near-CMC conditions (\cite{IM96},\cite{ACI08}), and more recently these near-CMC conditions have been relaxed by Holst, Nagy, and Tsogtgerel in \cite{HNT09} in the non-vacuum case, and by Maxwell in \cite{Maxwell09}. Their results used a Schauder fixed point argument to find a solution to the LCBY equations under the condition that the tensor $\sigma$ have sufficiently small norm.

These methods for finding solutions have been extended to other geometries including asymptotically Euclidean (\cite{DIMM13}), asymptotically hyperbolic (\cite{GS12}) and compact with boundary (\cite{Dilts13}, \cite{HMT13}). General results for the solvability of the LCBY equations on manifolds with asymptotically cylindrical or periodic geometries (which we collectively call \textit{geometries of cylindrical type}) appeared in 2012 with the pair of papers \cite{CM12} and \cite{CMP12} which analyzed the Lichnerowicz equation and vector equation separately. It was noted in the first of these that several known examples of black hole spacetimes admit CMC hypersurfaces with ends of cylindrical type, including the extreme Kerr solution. Solutions to the fully coupled LCBY equations were recently constructed on such manifolds in the paper \cite{Leach14}, using an adaptation of the Schauder fixed point technique.

Another approach to finding solutions of the LCBY equations on closed manifolds is the relatively recent result of Dahl, Gicquaud, and Humbert \cite{DGH11}. In that paper, the authors showed that in the absence of conformal Killing fields, so long as the mean curvature $\tau$ is non-vanishing and $\sigma \not\equiv 0$ for the non-negative Yamabe classes, then the system (\ref{origLich})-(\ref{origVect}) fails to have a solution only if there is a nontrivial solution $W$ to
\begin{equation}\label{limiteq}
\di LW = \alpha_0 \sqrt{\frac{n-1}{n}} |LW| \frac{d\tau}{\tau}
\end{equation}

\noindent for some $\alpha_0 \in (0,1]$. This is known as the \textit{limit equation}, for reasons that become clear in the proof of the theorem. We will give a summary of their proof in the next section.

This result has thus far proven to be quite adaptable to other (i.e. noncompact) geometries. Indeed, a limit equation result has very recently been shown to hold in the asymptotically Euclidean setting in \cite{DGI14} and in the asymptotically hyperbolic setting in \cite{GS12}, but difficulties arose in proving it in the compact with boundary case, c.f. \cite{Dilts13}. It is thus natural to wonder in which geometries such a result holds. In this paper, we prove that an analogous limit equation result holds on a manifold with asymptotically cylindrical ends, thus shedding more light on solvability of the constraint equations in such a geometry. We will also give an example of (conformal) initial data which does \textit{not} admit any solutions to (\ref{limiteq}), and thus necessarily admits a solution to the LCBY system (\ref{origLich})-(\ref{origVect}). We will do this by imposing a particular reasonable bound on the Ricci curvature.


\subsection{Notation and definitions} \label{notation}

We will assume that our manifold $(M,g)$ is asymptotically cylindrical. This means that there is some compact set $\mathcal{K} \subset M$ whose complement in $M$ decomposes as
\begin{equation*}\label{cyldecomp}
M\setminus \mathcal{K} = \bigsqcup^m _{\ell = 1} E_\ell.
\end{equation*}

\noindent Here, for each $\ell$, there is some closed $(n-1)$-manifold $N_\ell$ such that $E_\ell$ is diffeomorphic to the half-cylinder $\R_+ \times N_\ell$. Moreover, on each end the metric $g$ decays exponentially to the exactly cylindrical metric $dt^2 + \mr{g}_\ell$, where $\mr{g}_\ell$ is some metric on $N_\ell$ and $t$ is some $C^3$ coordinate function on the ends which is `radial' in the sense that there are positive constants $C_1 \leq C_2 \leq C_3$ such that
\begin{equation*}\label{radialt}
C_1 + t \leq C_2 + \mathrm{dist}_g(\cdot,\partial E_\ell) \leq C_3 + t.
\end{equation*}

\noindent We impose a similar decay condition on its derivatives, which can be made precise by requiring that on each end
\begin{equation*}\label{metricdecay}
|\mr\nabla^k [g - (dt^2 + \mr{g}_\ell)]| = \mathcal{O}(e^{-\omega t})
\end{equation*}

\noindent for some positive $\omega$ for all applicable $k$. Here $\mr\nabla$ is the covariant derivative associated to the exactly cylindrical metric. We will call a metric $\check{g}$ \textit{conformally asymptotically cylindrical} if we may write it in the form $\check{g} = w^{N-2} g$ where $g$ is an asymptotically cylindrical metric and $w$ is a positive function such that, on each end $E_\ell$, $w \rightarrow \mr{w}_\ell$ at the rate $\mathcal{O}(e^{-\omega t})$ along with its derivatives, where $\mr{w}_\ell$ is a smooth positive function on $N_\ell$. (``Smooth" in this paper will mean ``as smooth as the metric.") Note that this condition implies that the metric $\check{g}$ does not decay on the ends, and so metrics with conic or cusp singularities do not belong to the class of conformally asymptotically cylindrical metrics according to our definition.

We would like the class of conformal initial data $(g,\sigma,\tau)$ we consider to have certain nice asymptotic properties like the metric. We will thus restrict to \textit{tame} initial data, defined as follows.
\begin{defn}
Conformal initial data $(g,\sigma,\tau)$ is said to be tame if $g$ is (conformally) asymptotically cylindrical and both of the following are satisfied:
\begin{itemize}

\item $\tau^2 \geq \tau_0 ^2 > 0$ and $\tau^2 \rightarrow \mr{\tau}_\ell ^2$ on the end $N_\ell$ at the rate $\mathcal{O}(e^{-\omega t})$, where $\mr\tau_\ell ^2$ and $\tau_0 ^2$ are positive constants.

\item $|\sigma|^2 _g \rightarrow \mr{\sigma}^2$ at the rate $\mathcal{O}(e^{-\omega t})$ where $\mr{\sigma}^2 \not\equiv 0$ is some smooth nonnegative function on $N_\ell$, or any smooth function on $N_\ell$ if $R<0$ on the ends.

\end{itemize}
\end{defn}

\noindent Our main results, stated below, make the stronger assumption that each $\mr\tau_\ell^2$ is a constant.

Now that we have defined our admissible initial data, we need to define the function spaces in which we will perform our analysis. These will be defined as in \cite{Leach14}. For any $1 \leq p < \infty$, define the weighted $L^p$-Sobolev space $W^{k,p} _\delta$ to be the space of all functions which are finite with respect to the norm
\begin{equation}\label{sobolevnorm}
\|X\|_{W^{k,p} _\delta} = \bigg ( \sum_{j \leq k} \int_M |\nabla^j X|^p e^{-p \delta t} dV_g \bigg )^{1/p}.
\end{equation}

\noindent As in \cite{CMP12}, we will denote the function space $W^{k,2} _\delta$ by $H^k _\delta$. Next, assume without loss of generality that the radius of injectivity for $(M,g)$ is greater than 1. Then if $B_1(q)$ is the ball of radius 1 about the point $q \in M$, we define local and global weighted H\"older norms by
\begin{equation*}\label{localholder}
\|X\|_{k,\mu; B_1 (q)} = \sum^k_{i=0} \sup_{B_1 (q)} |\nabla^i X| + \sup_{x,y \in B_1 (q)} \dfrac{|\nabla^k X (x) - \nabla^k X(y)|}{d_g (x,y)^\mu}
\end{equation*} \begin{equation*}\label{globalholder}
\|X\|_{k,\mu} = \sup_{q \in M} \|X\|_{k,\mu; B_1 (q)},
\end{equation*}

\noindent and also the weighted H\"{o}lder norms by
\begin{equation*}\label{weightedholder}
\|X\|_{k,\mu,\delta} = \|e^{-\delta t} X\|_{k,\mu}.
\end{equation*}

\noindent The space of all functions (or tensor fields) for which the global H\"older norm is finite will be denoted by $C^{k,\mu}$, and the space of functions (or tensor fields) which are finite with respect to the weighted H\"older norm will be denoted by $C^{k,\mu} _\delta$. We similarly denote the standard sup-norm by $\|\cdot\|_\infty$ (or by $\|\cdot\|_0$ if restricted to $C^0$) and the corresponding weighted sup-norm by $\|\cdot\|_{\infty,\delta}$ (or $\|\cdot\|_{0,\delta})$.


\subsection{Preliminary results}

The conformal Killing operator $L$ which appears in the LCBY equations is given by
\begin{equation*}\label{confkillop}
(LX)_{ij} = \nabla_i X_j + \nabla_j X_i - \frac{2}{n} \di X g_{ij}.
\end{equation*}

\noindent This is merely the trace-free part of the deformation tensor associated to $X$. Any vector field in the nullspace of $L$ is called a \textit{conformal Killing field}, and the operator $\di L$ is known as the \textit{conformal vector Laplacian}. The mapping properties of this operator acting between weighted Sobolev spaces were studied extensively in \cite{CMP12}. The most important mapping property for us will be \cite[Thm 6.1]{CMP12}. Before we state the theorem, we first define $\mathscr{Y}_\ell$ to be the set of all globally bounded conformal Killing fields with respect to the \textit{exactly} cylindrical metric on $E_\ell$. If we then define a smooth cutoff function $\chi_\ell$ which vanishes outside of $E_\ell$ and is identically 1 where $t \geq 1$ on $E_\ell$, we set $\mathscr{Y} = \oplus \{ \chi_\ell Y : Y \in \mathscr{Y}_\ell \}$. The theorem then gives us the mapping properties of $\di L$ acting as an operator between weighted Sobolev spaces:

\begin{thm} \label{conformalVectorLaplacian}
Let $(M,g)$ be a Riemannian $n$-manifold with a finite number of ends which are asymptotically cylindrical. Suppose further that there are no global $L^2$ conformal Killing fields. Then there exists a number $\delta_* >0$ such that if $0<\delta<\delta_*$, then
  \[
  \di L: H^{k+2}_\delta(TM) \to H^k_\delta (TM)
  \] is surjective and the map
  \[
  \di L: H^{k+2}_{-\delta}(TM) \to H^k_{-\delta}(TM)
  \]is injective for every $k\geq0$. Moreover, if $0<\delta<\delta_*$, then for all $k\geq 0$ the map
  \[
  \di L : H^{k+2}_{-\delta}(TM) \oplus \mathscr{Y} \to H^k_{-\delta}(TM)
  \] is surjective with finite dimensional nullspace.
\end{thm}

\noindent This theorem is proven by using a parametrix construction found in \cite{Maz91} to build a generalized inverse $G: H^k _{-\delta} \rightarrow H^{k+2}_{-\delta} \oplus \mathscr{Y}$ (the latter space being a subspace of $H^{k+2} _\delta$) which satisfies $\di L \circ G = Id$. The space $\mathscr{Y}$ which appears in the third mapping property above is an artifact of the parametrix structure described in \cite{Maz91}. As shown in that paper, this parametrix is also a bounded map between weighted $L^p$-Sobolev spaces and weighted H\"older spaces, and so Theorem \ref{conformalVectorLaplacian} also holds when one considers $\di L$ as a map between such spaces. The observation that $\di L : W^{k+2,p} _{-\delta} \oplus \mathscr{Y} \rightarrow W^{k,p} _{-\delta}$ is surjective will be crucial in the analysis below.

The final analytical tools we will need are the Sobolev embedding theorems for function spaces on asymptotically cylindrical manifolds. In general, we will assume $p>n$ and that $\delta < \delta_*$, where $\delta_*$ is as in Theorem \ref{conformalVectorLaplacian}. We also need that the coefficients in the Lichnerowicz equation are regular enough to imply that the solution $\phi\in C^2$. This could be achieved by assuming the coefficients are in $W^{1,p}$ or in $C^{0,\beta}$ for some $1>\beta >0$, and thus we make the metric $W^{3,p}$ or $C^{2,\beta}$ to guarantee this regularity of the scalar curvature.

The standard Sobolev embedding theorems hold with a subscript $\delta$. In particular, we have the following theorem.

\begin{thm}\label{embeddings}
  Assume $t\in C^k$ with bounded derivatives up to $k$-th order.
  \begin{itemize}
    \item If $l-\frac{n}{q} = k-\frac{n}{p}$ and $k\geq l$, then $W^{k,p}_\delta(M) \subseteq W^{l,q}_\delta(M)$, with a continuous embedding.
    \item If $k- \frac{n}{p} = r+\alpha$, then $W^{k,p}_\delta(M) \subseteq C^{r,\alpha}_\delta(M)$, with a continuous embedding.
    \item If $k-\frac{n}{p} > r$, then the embedding $W^{k,p}_{\delta}(M) \hookrightarrow C^r_{\delta'}(M)$ is compact for any $\delta'>\delta$.
  \end{itemize}
\end{thm}
\begin{proof}
The first two parts are proven essentially the same way, so we only include the first. Suppose $f \in W^{k,p}_\delta(M)$. By the definition of the norms and since $\nabla^k t$ are bounded, it is clear that
\[
\|f e^{-\delta t} \|_{W^{l,q}(M)} \leq C\|f\|_{W^{l,q}_\delta(M)}.
\] A similar statement holds in the other direction, but is slightly harder to prove. We want
\[
\|f\|_{W^{k,p}_\delta} \leq C \|f e^{-\delta t}\|_{W^{k,p}}.
\] Suppose not. Then there exists a sequence $f_i$ such that $\|f_i\|_{W^{k,p}_\delta} = 1$ but $\|f_i e^{-\delta t}\|_{W^{k,p}} \leq 1/i$. This condition implies that
\[
\int |f_i|^p e^{-\delta p t} \to 0
\] since each term in (\ref{sobolevnorm}) must got to zero separately. We also have
\[
\int |\nabla (f_i e^{-\delta t})|^p \to 0.
\] Note that
\[
\left(\int |\nabla (f_i e^{-\delta t})|^p\right)^{1/p} \geq \left(\int |\nabla f_i|^p e^{-\delta p t}\right)^{1/p} - \left(\int|\delta \nabla t|^p |f_i|^p e^{-\delta p t}\right)^{1/p}
\] and since the first and last integrals goes to zero (since $|\nabla t|$ is bounded), so must the middle one. We can continue this process to the $k$-th derivative term. Thus $\int |\nabla^j f_i|^p e^{-\delta p t} \to 0$ for any $0\leq j\leq k$, which contradicts our assumption. Thus the inequality holds.

The embedding theorem is now clear since
\[
\|f\|_{W^{k,p}_\delta} \leq C \|f e^{-\delta t}\|_{W^{k,p}} \leq C \|f e^{-\delta t}\|_{W^{l,q}} \leq C \|f\|_{W^{l,q}_\delta}.
\]

We next prove that $C^{0,\alpha}_{\delta} \hookrightarrow C^0_{\delta'}$ is a compact embedding for $\delta'>\delta$. This combined with the second statement gives the third statement. Suppose we have a uniformly bounded family $\{f_i\}\subset C^{0,\alpha}_{\delta}$. Let $K_j$ be an exhaustion of $M$ by compact sets. A standard argument using Arzela-Ascoli gives that there is a convergent subsequence of $f_i$ converging in $C^0$ on $K_j$. Take the subsequence for $K_1$ and reduce it further to subsequence that converges on $K_2$, etc. We redefine $f_1$ to be the first term of the subsequence converging on $K_1$, $f_2$ to be the second term of the subsubsequence converging on $K_2$, etc. Thus this new sequence $f_i$ converges on any compact set of $M$. This is just the standard diagonal subsequence.

Pick $\epsilon>0$. Note that we have $f_i < C e^{\delta t}$ for some uniform $C$. For any $\epsilon$ there is a $t$ large enough such that $C e^{\delta t} < \frac{\epsilon}{2} e^{\delta' t}$ since $\delta'>\delta$; we call the set where this is true $U$. Since $M\setminus U$ is compact, $f_i$ converges in $C^0$ on $M\setminus U$. In particular, since $e^{\delta' t}$ is bounded on $M\setminus U$, we have, for large enough $i,j$, that $|f_i(x) - f_j(x)|e^{-\delta' t} < \epsilon$ for any $x\in M\setminus U$ as well. Thus $f_i$ converges in $C^0_{\delta'}$ in $M\setminus U$. For $x\in U$ we have $|f_i(x) - f_j(x)|e^{-\delta' t} <\epsilon$ by definition. Thus, for any $\epsilon>0$, $|f_i(x) - f_j(x)|e^{-\delta' t} <\epsilon$ for $i,j$ large enough, and so $f_i$ converges everywhere in $C^0_{\delta'}$.
\end{proof}


\section{Summary of main results}

In what follows we will first use a strategy very similar to that in \cite{Leach14} to prove that we may find a solution to the subcritical LCBY equations:

\begin{thm}\label{subcritthm} Let $(g,\sigma,\tau)$ be tame conformal data on a complete Riemannian $n$-manifold $(M,g)$ with a finite number of asymptotically cylindrical ends. Assume that there are no global $L^2$ conformal Killing fields. If $\sigma \in W^{1,p}$, $\tau - \mr \tau \in W^{1,p} _{-\delta}$, and the scalar curvature $R$ satisfies $c_n R + b_n \tau^2 > 0$, then the subcritical LCBY equations
\begin{equation}\label{subcrit1}
\Delta_g \phi - c_n R_g \phi = b_n \tau^2 \phi^{N-1} - c_n |\sigma + LW|_g ^2 \phi^{-N-1}
\end{equation} \begin{equation}\label{subcrit2}
\di LW = \frac{n-1}{n} \phi^{N-\epsilon} d\tau
\end{equation}

\noindent have a solution $(\phi,W)$ with $\phi-\mr\phi \in W^{3,p}_{-\nu}$ for some $\nu>0$ and $\phi>0$ and $W\in W^{2,p}_{-\delta} \oplus \mathscr{Y}$. The function $\mr{\phi}$ is defined below in \eqref{eq:redlich}.
\end{thm}

\noindent Once we know that the system (\ref{subcrit1})-(\ref{subcrit2}) can always be solved, we next establish that such a solution $\phi$ can be uniformly (i.e. independent of $\epsilon$) bounded above by the energy $\|LW\|_{L^2}$. Thus to any sequence of sufficiently small positive numbers $\epsilon_i \rightarrow 0$, we can associate a sequence of energies
\begin{equation}
\gamma_i = \int_M |LW_i|^2
\end{equation}

\noindent where $(\phi_i,W_i)$ is some solution to (\ref{subcrit1})-(\ref{subcrit2}) with $\epsilon = \epsilon_i$. It is natural to ask what happens to the sequence $\{\gamma_i\}$ as $i \rightarrow \infty$. It turns out, as was shown by \cite{DGH11} in the closed manifold case, that this limit gives us valuable information about the solvability of the LCBY equations (\ref{origLich})-(\ref{origVect}). In particular, if this sequence stays bounded, then there exists a solution to the LCBY equations with the same hypotheses on conformal data as in Theorem \ref{subcritthm}. On the other hand, if $\gamma_i \rightarrow \infty$, we show that the limit equation (\ref{limiteq}) admits a nontrivial solution. The precise statement of this result is our main theorem:
\begin{thm}\label{mainTheorem}
  Let $(g,\sigma,\tau)$ be conformal data on an asymptotically cylindrical manifold satisfying the conditions of Theorem \ref{subcritthm}. Then at least one of the following is true:
  \begin{itemize}
    \item The system \eqref{origLich}-\eqref{origVect} admits a solution $(\phi,W)$ with regularity as in Theorem \ref{subcritthm}. Also, the set of these solutions is compact.

    \item There exists a non-zero solution $W \in W^{2,p}_{-\delta}\oplus \mathscr{Y}$ of the limit equation
  \[
  \di LW = \alpha_0 \sqrt{\frac{n-1}{n}} |LW| \frac{d\tau}{\tau}
  \] for some $\alpha_0 \in (0,1]$ such that $|LW| \leq C e^{-\delta t}$ for some $C$ independent of $\phi_i$, $W_i$ and $W$.
  \end{itemize}
\end{thm}

Note that we have no reason to suspect that this result is a dichotomy. That is, it may hold that the LCBY equations and the limit equation \textit{both} admit nontrivial solutions. The utility of the main theorem is that we may show the existence of solutions to the LCBY equations whenever we can show that the limit equation (\ref{limiteq}) admits no solutions. The idea behind the proof is that given any sequence $\{\gamma_i\}$ as above, if the limit equation admits no solutions then the sequence stays bounded, which in turn implies that the LCBY equations admit a solution (these facts are proven in section 4 below). From the perspective of one trying to understand the constraint equations on asymptotically cylindrical manifolds, this result is only useful if we can find examples of conformal data which admit no solutions to the limit equation. In section 5, we show that this set is nonempty by proving the following corollary.
\begin{cor}
Let $(g,\sigma,\tau)$ be conformal data on an asymptotically cylindrical manifold satisfying the conditions of Theorem \ref{subcritthm}, and suppose $\Ric \leq (c_1\chi^2 - c_2 e^{-2\mu t}) g$ for some constants $c_i, \mu >0$ and a smooth compactly supported bump function $\chi$. Then there is some $C > 0$ such that if
\[
\left\|\frac{d\tau}{\tau} \right\|_{C^0_{-\mu}} < C,
\] there is a solution to the LCBY equations (\ref{origLich})-(\ref{origVect}).
\end{cor}

Finally, in section 6 we discuss how these results can be extended to the case where the metric $g$ is only assumed to be \textit{conformally} asymptotically cylindrical.


\section{Global barriers} \label{barriers}

In this section we prove Theorem \ref{subcritthm} by constructing a global supersolution and a global subsolution (defined below) with desirable asymptotics. Throughout this section, we assume for notational simplicity that our manifold has a single asymptotically cylindrical end. It is easy to see that this assumption results in no loss of generality, for the extension of our analysis to multiple ends is completely trivial. We shall make use of the following lemma, whose proof is nearly identical to one found in \cite{Leach14}:

\begin{lem} \label{dwbd}
Let $\phi$ be a bounded continuous function on a manifold $M$ which does not admit any $L^2$ conformal Killing fields and $d \tau \in C^0 _{-\delta'} (TM)$ for some positive $\delta' < \delta_*$. If $W_\phi$ is the solution of the subcritical momentum constraint equation associated to a conformally asymptotically cylindrical metric, then for any $\delta < \delta'$, there exists some constant $K$ which does not depend on $\phi$ or $\tau$ such that the following pointwise estimate holds:
\begin{equation}
|LW_\phi | \leq K \|d \tau\|_{0,-\delta'} \|\phi\|_0 ^{N-\epsilon} e^{-\delta t}.
\end{equation}
\end{lem}

\noindent We define the constant $K_\tau = K \|d \tau\|_{0,-\delta'}$ for brevity in our analysis below. Note in particular that the lemma implies that $|\sigma + LW|_g ^2 \rightarrow \mr{\sigma}^2$ on the end. Hence we expect a solution which has an asymptotic limit on the cylindrical end to approach a solution to the \textit{reduced Lichnerowicz equation}, which is
\begin{equation} \label{eq:redlich}
\Delta_h \mr{\phi} - c_n \mr{R} \mr{\phi} = b_n \mr{\tau}^2 \mr{\phi}^{N-1} - c_n \mr{\sigma}^2 \mr{\phi}^{-N-1}.
\end{equation}

\noindent Given tame data $(g,\sigma,\tau)$, this equation admits a constant supersolution. Moreover, one easily checks that if $\mr\phi_1$ is a positive solution to the equation
\begin{equation*}
\Delta_h u - (c_n \mr{R} + b_n \tau^2) u = -\mr\sigma^2,
\end{equation*}

\noindent then $\rho\mr\phi_1$ is a subsolution of the reduced equation for any sufficiently small $\rho>0$. Hence (\ref{eq:redlich}) admits a solution by Proposition A.4 in \cite{CM12}. We note that this is the only place where we use the requirement that $\sigma \not\equiv 0$ unless $R<0$ on the ends. We extend this solution to a smooth function on all of $M$ which we continue to call $\mr\phi$.

Next, we define the nonlinear operator $\Lich_\phi$ on $C^2$ by
\begin{equation}
\Lich_\phi (u) := \Delta_g u - c_n R u - b_n \tau^2 u ^{N-1} + c_n |\sigma + LW_\phi|^2 _g u ^{-N-1}
\end{equation}

\noindent where $W_\phi$ is the solution of the subcritical momentum constraint equation (\ref{subcrit2}) (such a solution exists by Theorem 6.1 in \cite{CMP12}). We call $\phi_+$ a \textit{global supersolution} if for any $0 < \phi \leq \phi_+$, we have $\Lich_\phi (\phi_+) \leq 0$. If $\phi_+$ is of lower regularity so that this inequality only holds in the weak sense, we call it a \textit{weak global supersolution}. Given a (weak) global supersolution $\phi_+$, we define an associated \textit{(weak) global subsolution} $\phi_- >0$ which, for any $\phi_- \leq \phi \leq \phi_+$, satisfies $\Lich_\phi (\phi_-) \geq 0$. Having constructed a pair of weak global sub/supersolutions $\phi_- \leq \phi_+$ for the system (\ref{subcrit1})-(\ref{subcrit2}) which both approach $\mr{\phi}$ asymptotically on the end, one argues as in \cite{Leach14} and applies the Schauder fixed point theorem to produce a solution to the system.

In what follows, we construct a pair of weak global sub/supersolutions for the subcritical LCBY equations. The construction below is independent of the Yamabe invariant of the background metric, though we assume as in \cite{CM12} that $c_n R + b_n \tau^2 > 0$ everywhere. If one drops this assumption in the case where the Yamabe invariant is positive, the construction of a pair of weak global sub/supersolutions proceeds almost exactly as in \cite{Leach14}.


\subsection{Global supersolution}

We first show that a sufficiently large constant $B$ is a global supersolution for the subcritical LCBY system. Suppose $0 < \phi \leq B$. Using Lemma \ref{dwbd} we have (dropping all subscripts ``$g$'')
\begin{align*}
\Lich_\phi (B) &= -c_n R B - b_n \tau^2 B^{N-1} + c_n |\sigma + LW_\phi|^2 B^{-N-1} \\
&\leq -c_n R B - b_n \tau^2 B^{N-1} + c_n B^{-N-1}(|\sigma|^2 + 2 |\sigma||LW_\phi| +|LW_\phi|^2) \\
&\leq -c_n R B - b_n \tau^2 B^{N-1} + C_1 B^{-N-1} + C_2 B^{-1-\epsilon} + C_3 B^{N-1-2\epsilon}
\end{align*}

\noindent where we have used Lemma \ref{dwbd} and defined $C_1 = c_n \sup_M |\sigma|^2$, $C_2 = 2 c_n \sup_M |\sigma| K_\tau$, and $C_3 = c_n K_\tau^2$. Since $\tau^2 \geq \tau_0 ^2 > 0$, we see that this final expression is negative for sufficiently large $B$ and hence such a constant is a global supersolution. Notice that subcriticality was crucial in establishing this fact.

We would like to construct a supersolution which approaches the function $\mr\phi$ on the ends asymptotically. This way the solution we obtain from the fixed point theorem will also approach $\mr\phi$. For this we prove the following proposition whose proof is similar to Proposition 3.5 in \cite{Leach14} and Theorem 4.3 in \cite{CM12}:

\begin{prop} \label{limitsup}
Choose some positive $\nu < \delta/(2N+2)$ so small that $\Lich_0 (\mr\phi) = \mathcal{O}(e^{-2\nu t})$. With the initial data above, there exists some $T > 0$ such that the function $\mr\phi (1 + b e^{-\nu t})$ is a global supersolution for any sufficiently large $b$ on the set where $t \geq T$.
\end{prop}

\noindent Once we establish this proposition, one easily sees that we can choose $b$ so large that $\mr\phi (1 + b e^{-\nu t}) > B$ whenever $t \leq T$ and therefore $\inf(B,\mr\phi (1 + b e^{-\nu t}))$ is continuous and thus a weak global supersolution with the desired asymptotics.

\proof Define a new metric $\tilde{g} = \mr\phi ^{N-2}g$, and denote with a tilde all operators and quantities associated to this metric. The Lichnerowicz equation is well known to be conformally covariant, which means that $\mr\phi u$ is a solution to
\begin{equation*}
\Delta_g (\mr\phi u) - c_n R_g \mr\phi u = b_n \tau^2 (\mr\phi u)^{N-1} - c_n |\sigma|_g^2 (\mr\phi u)^{-N-1}
\end{equation*}

\noindent if and only if $u$ is a solution to
\begin{equation*}
\Delta_{\tilde{g}} u - c_n R_{\tilde{g}} u = b_n \tau^2 u^{N-1} - c_n |\mr\phi^{-2} \sigma |_{\tilde{g}} ^2 u^{-N-1}.
\end{equation*}

\noindent It follows from this pointwise property that for any $\phi \leq \mr\phi (1 + b e^{-\nu t})$, we have
\begin{align*}
\Lich_\phi (\mr\phi (1 + b e^{-\nu t}))&= (\Delta - c_n R)(\mr\phi (1 + b e^{-\nu t})) - b_n \tau^2 (\mr\phi (1 + b e^{-\nu t}))^{N-1} \\
&\quad + c_n |\sigma + LW_\phi|^2 (\mr\phi (1 + b e^{-\nu t}))^{-N-1} \\
&= (\tilde \Delta - c_n \tilde R )(1 + b e^{-\nu t}) - b_n \tau^2 (1 + b e^{-\nu t})^{N-1} \\
&\quad + c_n |\mr\phi^{-2} (\sigma + LW_\phi)|^2 _{\tilde{g}} (1 + b e^{-\nu t})^{-N-1} \\
&= (\tilde \Delta - c_n \tilde R )(1 + b e^{-\nu t}) - b_n \tau^2 (1 + b e^{-\nu t})^{N-1} \\
&\quad + c_n |\tilde\sigma + \mr\phi^{-2} LW_\phi|^2 _{\tilde{g}} (1 + b e^{-\nu t})^{-N-1}.
\end{align*}

\noindent Next observe by the covariance of the conformal Laplacian that
\begin{equation} \label{conflap}
(\tilde \Delta - c_n \tilde R)u = \mr\phi^{-(N-1)} (\Delta - c_n R)(\mr\phi u)
\end{equation}

\noindent for any function $u\in C^2$, and so taking $u \equiv 1$ gives us
\begin{align*}
-c_n \tilde R &= \mr\phi^{-(N-1)} (b_n \tau^2 \mr\phi^{N-1} - c_n |\sigma|^2 \mr\phi^{-N-1} + \mathcal{O}(e^{-2 \nu t})) \\
&= b_n \tau^2 - c_n |\tilde\sigma|^2 _{\tilde{g}} + s
\end{align*}

\noindent where $s$ is some function satisfying $|s| \leq C e^{-2 \nu t}$. We thus see that $\Lich_\phi (\mr\phi (1 + b e^{-\nu t}))$ is given by the expression
\begin{multline*}
(\tilde \Delta + b_n \tau^2 - c_n |\tilde\sigma|^2 _{\tilde{g}} + s) (1 + b e^{-\nu t}) - b_n \tau^2 (1 + b e^{-\nu t})^{N-1}  \\
+ c_n |\tilde\sigma + \mr\phi^{-2} LW_\phi|^2 _{\tilde{g}} (1 + b e^{-\nu t})^{-N-1}
\end{multline*}

\noindent which, after adding and subtracting $(N-1)b_n \tau^2 b e^{-\nu t}$, we may rewrite as
\begin{multline*}
 \left(\tilde \Delta - \left(c_n |\tilde\sigma|_{\tilde{g}}^2 + (N-2) b_n \tau^2 - s\right)\right) (b e^{-\nu t}) + b_n \tau^2 \left((N-1)b e^{-\nu t}\right. \\\left.+ 1 - (1 + b e^{-\nu t})^{N-1}\right)
+ c_n |\tilde\sigma + \mr\phi^{-2} LW_\phi|^2 _{\tilde{g}}(1 + b e^{-\nu t})^{-N-1} - c_n |\tilde\sigma|^2 _{\tilde{g}} + s.
\end{multline*}

\noindent We now make a few observations about this expression. First, since $\tau_0 > 0$, we see that for some large $T_0$, we have
\begin{equation}\label{h}
h := c_n |\tilde\sigma|_{\tilde{g}} ^2 + (N-2) b_n \tau^2 - s \geq c > 0
\end{equation}

\noindent for all $t \geq T_0$. Note too that the expression $(N-1)b e^{-\nu t} + 1 - (1 + b e^{-\nu t})^{N-1}$ is negative, as may be seen by differentiating the function $r(x) = (N-1)x + 1 - (1 + x)^{N-1}$ and observing that $r(0) = 0$.

Having chosen $T_0$ as above, note that for any $\rho > 0$ satisfying $4 \rho^2 < c$, if $\nu < \rho$ we have $\tilde \Delta e^{-\nu t} \leq 2 \rho^2 e^{-\nu t}$ and hence $(\tilde \Delta - h) e^{-\nu t} \leq - 2 \rho^2 e^{-\nu t}$. From this we conclude that, for $t \geq T_0$, $\Lich_\phi (\mr\phi (1 + b e ^{-\nu t}))$ is bounded above by
\begin{multline}
-2 b\rho^2 e^{-\nu t} + s + b_n \tau^2 b^{N-1} f(b,t) + 2 c_n |\tilde\sigma|_{\tilde{g}} |\mr\phi^{-2} LW_\phi|_{\tilde{g}} b^{-N-1} e^{(N+1)\nu t} \notag \\
+ c_n |\mr\phi^{-2} LW_\phi|^2_{\tilde{g}} b^{-N-1} e^{(N+1)\nu t}
\end{multline}\

\noindent where we have simultaneously used the basic inequalities
\begin{align}
|\tilde\sigma + \mr\phi^{-2} LW |_{\tilde{g}} ^2 &\leq |\tilde\sigma|_{\tilde{g}} ^2 + 2|\tilde\sigma|_{\tilde{g}}|\mr\phi^{-2} LW |_{\tilde{g}} + |\mr\phi^{-2} LW|_{\tilde{g}} ^2 \\
(1 + b e^{-\nu t})^{-N-1} &\leq \min(1,b^{-N-1} e^{(N+1)\nu t})
\end{align}

\noindent and set
\begin{equation*}
f(b,t) = \frac{1}{b^{N-1}} + \frac{N-1}{b^{N-2}} e^{-\nu t} - \left(\frac{1}{b} + e^{-\nu t}\right)^{N-1}.
\end{equation*}

\noindent Noting that $\tilde L W = \mr \phi ^{-(N-2)} LW$, one applies Lemma \ref{dwbd} and the fact that $\phi \leq \mr\phi (1+be^{-\nu t})$ to find that the bound for $\Lich_\phi (\mr\phi (1 + be^{-\nu t})$ we obtained can itself be bounded above by
\begin{multline}
-2 b\rho^2 e^{-\nu t} + s + b_n \tau ^2 b^{N-1} f(b,t)\\
+ k_1 ||\mr\phi (1 + b e^{-\nu t}) ||^{N-\epsilon} _\infty b^{-N-1} e^{((N+1)\nu - \delta) t}  \\
+ k_2 ||\mr\phi (1 + b e^{-\nu t})||^{2N - 2\epsilon} b^{-N-1} e^{((N+1)\nu -2\delta) t}. \label{expr2}
\end{multline}

\noindent where $k_1 = 2 c_n K_\tau \|\tilde \sigma \|_\infty \|\mr \phi \|_\infty ^{N-4}$ and $k_2 = c_n K_\tau ^2 \|\mr\phi\|_\infty ^{2N-8}$ are constants. Finally, when $b \gg 1$ we may factor out the $b$ from the expression $||\mr\phi (1 + b e^{-\nu t})||_\infty$ to find that (\ref{expr2}) is bounded above by
\begin{multline}\label{mustbd}
-2 b\rho^2 e^{-\nu t} + s + b_n \tau ^2 b^{N-1} f(b,t) + k' _1 b^{-1-\epsilon} e^{((N+1)\nu - \delta) t}  \\
+ k' _2 b^{N -1 - 2\epsilon} e^{((N+1)\nu -2\delta) t}
\end{multline}

\noindent where $k' _1 = (2\|\mr\phi\|_\infty) ^{N-\epsilon} k_1$ and $k_2 = (2 \|\mr\phi\|_\infty) ^{2N-2\epsilon} k_2$. We may thus choose some $b_0$ so large that, for all $b \geq b_0$, we have
\begin{equation}
- b_n \tau_0 ^2 b^{N-1} + k' _1 b^{-1 - \epsilon} + k' _2 b^{N-1-2\epsilon} < 0.
\end{equation}

\noindent Now to prove the proposition, we show that there is some choice of $T > 0$ such that (\ref{mustbd}) is negative for any $t \geq T$ and $b \geq b_0$. Clearly there is some $T_1 \geq T_0$ such that $-2 \rho^2 e^{-\nu t} + s < 0$ for all $t \geq T_1$, so our task is reduced to proving that the sum of the final three terms in (\ref{mustbd}) is negative for such a choice of $t$ and $b$. The analysis of these terms differs slightly depending on whether $N \geq 3$ (i.e. $n \leq 6$) or $N < 3$, so we consider these cases separately.

First suppose $N \geq 3$. For any fixed $t$, we see that $f(b,t) \rightarrow -e^{-(N-1)\nu t}$ as $b \rightarrow \infty$. On the other hand, we find by differentiating that $f$ is increasing in $b$, so we thus conclude that $f(b,t) < -e^{(N-1)\nu t}$ for \textit{all} $b$ and $t$. The final three terms in (\ref{mustbd}) are thus bounded above by
\begin{align*}
-b_n \tau_0 ^2 b^{N-1} e^{-(N-1) \nu t} &+ k' _1 b^{-1-\epsilon} e^{((N+1)\nu - \delta) t} + k' _2 b^{N -1 - 2\epsilon} e^{((N+1)\nu -2\delta) t} \\
&\leq e^{-(N-1) \nu t}(-b_n \tau_0 ^2 b^{N-1}  + k' _1 b^{-1-\epsilon}  + k' _2 b^{N -1 - 2\epsilon}) \\
&<0
\end{align*}

\noindent where we have used our smallness condition on $\nu$ in the first inequality and the fact that $b \geq b_0$ in the second. This completes the proposition in the case where $N \geq 3$.

The case in which $N < 3$ is not much more difficult. The key difference is that we now have $\partial f / \partial b < 0$. Note first that for any fixed $b$, one computes the limit (using a Taylor expansion, for example)
\begin{equation}
\lim_{t\rightarrow \infty} \frac{f(b,t)}{e^{-2 \nu t}} = -\frac{b^{3-N} (N-1)(N-2)}{2},
\end{equation}

\noindent so we may thus choose some $b_1 \geq b_0$ such that this limit is less than $-2$ for all $b \geq b_1$. Hence there exists some $T_2 \geq T_1$ such that $f(b_1,t) < -e^{-2\nu t}$ for all $t \geq T_2$. But since $\partial f / \partial b < 0$, we see that $f(b,t) < -e^{-2\nu t}$ for \textit{any} $b \geq b_1$ and hence for any such $b$ the final three terms in (\ref{mustbd}) are bounded above by
\begin{align*}
-b_n \tau_0 ^2 b^{N-1} e^{-2 \nu t} &+ k' _1 b^{-1-\epsilon} e^{((N+1)\nu - \delta) t} + k' _2 b^{N -1 - 2\epsilon} e^{((N+1)\nu -2\delta) t} \\
&\leq e^{-2 \nu t}(-b_n \tau_0 ^2 b^{N-1}  + k' _1 b^{-1-\epsilon}  + k' _2 b^{N -1 - 2\epsilon}) \\
&<0
\end{align*}

\noindent for any $t \geq T_2$. This proves the case and the proposition. \qed


\subsection{Global subsolution}

Based on our construction of a global supersolution, we may suspect that the function $\mr\phi (1 - a e^{-\nu t})$ will provide a global subsolution far out on the end. We will show this to be the case. First observe, again by the conformal covariance of the Lichnerowicz equation, that for any $\phi \geq \mr\phi (1 - a e^{-\nu t})$ we have that $\Lich_\phi (\mr\phi (1 - a e^{-\nu t}))$ is given by
\begin{align*}
&(\tilde\Delta -c_n \tilde R)(1 - a e^{-\nu t}) - b_n \tau^2 (1 - a e^{-\nu t})^2 \\
&\qquad + c_n|\tilde\sigma + \mr\phi^{-2} LW_\phi|_{\tilde{g}} ^2 (1-a e^{-\nu t})^{-N-1} \\
&= (\tilde\Delta + b_n \tau^2 - c_n |\tilde\sigma|_{\tilde{g}} ^2 + s)(1 - a e^{-\nu t}) - b_n \tau^2 (1-a e^{-\nu t})^{N-1} \\
&\qquad + c_n |\tilde\sigma + \mr\phi^{-2} LW_\phi|_{\tilde{g}} ^2 (1-a e^{-\nu t})^{-N-1} \\
&= -a \tilde\Delta e^{-\nu t} + s (1 - a e^{-\nu t}) + (b_n \tau^2 - c_n |\tilde\sigma|_{\tilde{g}} ^2) (1-ae^{-\nu t}) - b_n \tau^2 (1-ae^{-\nu t})^{N-1} \\
&\qquad + c_n |\tilde\sigma + \mr\phi^{-2} LW_\phi|_{\tilde{g}} ^2 (1-a e^{-\nu t})^{-N-1}.
\end{align*}

\noindent As noted in \cite{CM12}, given any positive number $\mu < 1$, there is some $\mu' >0$ such that for all $y \in [\mu,1)$,
\begin{equation}
\frac{y-y^{N-1}}{1-y} \geq \mu'.
\end{equation}

\noindent Hence, fixing $a$ and choosing $t$ so large that $\mu < 1 - a e^{-\nu t} < 1$, we find that
\begin{equation}
b_n \tau^2 (1-ae^{-\nu t}) - b_n \tau^2 (1-ae^{-\nu t})^{N-1} \geq b_n \tau^2 a \mu' e^{-\nu t}.
\end{equation}

\noindent We thus conclude that, for such a choice of $t$, $\Lich_\phi (\mr\phi (1-ae^{-\nu t}))$ is bounded below by
\begin{multline*}
-a \tilde\Delta e^{-\nu t} + s(1-ae^{-\nu t}) + b_n \tau_0 ^2 a \mu' e^{-\nu t} - c_n |\tilde\sigma |^2 _{\tilde{g}} (1 - a e^{-\nu t}) \\
+ c_n |\tilde\sigma + \mr\phi^{-2} LW_\phi|_{\tilde{g}} ^2 (1-ae^{-\nu t})^{-N-1}.
\end{multline*}

\noindent Now using that $(1-ae^{-\nu t})^{-N-1} > 1 > 1 - ae^{-\nu t}$ and that $|\tilde\sigma + \mr\phi^{-2} LW_\phi|_{\tilde{g}} ^2 \geq (|\tilde\sigma|_{\tilde{g}}-\mr\phi^{-2} |LW_\phi|_{\tilde{g}})^2$, we see that the previous expression is bounded below by
\begin{multline*}
 -a \tilde\Delta e^{-\nu t} + s(1-ae^{-\nu t}) + b_n \tau_0 ^2 a \mu' e^{-\nu t} + c_n |\mr\phi^{-2} LW_\phi|_{\tilde{g}}(|\mr\phi^{-2} LW_\phi |_{\tilde{g}} - 2 |\tilde\sigma|_{\tilde{g}}) \\
 + c_n |\tilde\sigma|_{\tilde{g}} ^2 [(1-a e^{-\nu t})^{-N-1} - (1-ae^{-\nu t})] \\
\geq -a \tilde\Delta e^{-\nu t} + s(1-ae^{-\nu t}) + b_n \tau_0 ^2 a \mu' e^{-\nu t} -2c_n |\mr\phi^{-2} LW_\phi|_{\tilde{g}}|\tilde\sigma|_{\tilde{g}}.
\end{multline*}

\noindent At this point we use the fact that $\phi \leq \phi_+ := \mr\phi (1+ae^{-\nu t})$ and Lemma \ref{dwbd} to conclude
\begin{align*}
\Lich_\phi (\mr\phi (1 - ae^{-\nu t})) &\geq -a \tilde\Delta e^{-\nu t} + s(1-ae^{-\nu t}) + b_n \tau_0 ^2 a \mu' e^{-\nu t} - k_1 \|\phi_+\|_\infty  ^{N-\epsilon} e^{-\delta t},
\end{align*}

\noindent where the constant $k_1$ is defined in the previous section. Now there exists a constant $C >0$ such that $\tilde\Delta e^{-\nu t} \leq C \nu^2 e^{-\nu t}$, so the previous expression is bounded below by
\begin{equation} \label{fin}
-a C \nu^2 e^{-\nu t} + s(1-ae^{-\nu t}) + b_n \tau_0 ^2 a \mu' e^{-\nu t} - C' e^{-\delta t}.
\end{equation}

\noindent We thus see that if we choose $\nu \leq \tau_0 \sqrt{b_n \mu'/C}$, the total contribution of the first and third term in positive. Since the second term decays like $\mathcal{O}(e^{-2\nu t})$, we conclude that, with this choice of $\nu$, for some $T' > 0$ the expression \ref{fin} is positive for all $t \geq T'$.

Our next objective is then to find a global subsolution on the compact piece $\mathcal{K} = \{t \leq T'\}$ which is positive yet sufficiently small on the boundary of $\mathcal{K}$. Calling such a function $\eta$, the function $\sup(\eta,\mr\phi (1+ae^{-\nu t}))$ is then continuous and thus a weak global subsolution. To accomplish this, we merely define a slightly larger compact set $\mathcal{K'} = \{t\leq T''\}$, where $T'' > T'$ is so large that $ae^{-\nu t} < 1/2$ for all $t \geq T''$, and solve the Dirichlet problem
\begin{equation}
\begin{cases}
(\Delta - c_n R - b_n \tau^2)\eta = 0 \\
\eta|_{\partial \mathcal{K}'} = \frac{1}{2} \inf(1,\inf_M \mr\phi)
\end{cases}.
\end{equation}

The function $\eta$ is positive and less than 1 on the boundary of $\mathcal{K}'$ and hence on all of $\mathcal{K}'$ by the maximum principle. One easily checks that $\eta$ is a global subsolution on $\mathcal{K}'$, and it is less than $\mr\phi (1-ae^{-\nu t})$ near the boundary of $\mathcal{K}'$ by our choice of $T''$. Therefore, if we extend $\eta$ to be zero identically outside of $\mathcal{K}'$, we conclude that $\sup(\eta,\mr\phi (1+ae^{-\nu t}))$ is a weak global subsolution. We note that this is the only place where we used the condition that $c_n R + b_n \tau^2 >0$.


\subsection{Continuity of the solution maps}

Having constructed global sub/supersolutions of the subcritical system (\ref{subcrit1})-(\ref{subcrit2}), one finds a solution to this system with an application of the Schauder fixed point theorem as in \cite{HNT09}, \cite{Maxwell09}, or \cite{Leach14}. The fixed point theorem we need, a proof of which can be found in \cite{Istratescu81}, is the following:

\begin{thm}\label{Schauder}
Let $X$ be a Banach space, and let $U \subset X$ be a non-empty, convex, closed, bounded subset. If $T: U \rightarrow U$ is a compact operator, then there exists a fixed point $u \in U$ such that $T(u) = u$.
\end{thm}

\noindent To apply this theorem in the present context, we look for a solution $(\phi, W)$ where $\phi = \mr\phi + \psi$ and $\psi \in W^{3,p} _{-\nu}$. In this way we think of the subcritical system as having a solution $(\psi, W) \in W^{3,p} _{-\nu} (M) \times (W^{2,p} _{-\delta} (TM) \oplus \mathscr{Y})$, and $\psi$ shall be found as a fixed point of a particular function on the set
\begin{equation}
U = \{ \psi \in L^\infty _{-\nu'} : \phi_- - \mr\phi \leq \psi \leq \phi_+ - \mr\phi \}
\end{equation}

\noindent where we choose some positive $\nu' < \nu$. The set $U$ clearly meets all the criteria of Theorem \ref{Schauder} as a subset of the space $L^\infty _{-\nu'}$.

Let $\mathcal{W}_\epsilon : U \rightarrow W^{2,p} _{-\delta} (TM) \oplus \mathscr{Y}$ be the map which sends $\psi \in U$ to the vector field $G(n^{-1}(n-1) (\mr\phi + \psi)^{N-\epsilon} d\tau)$, where $G$ is the bounded generalized inverse defined above for the operator $\di L$. One can easily see that the map $\mathcal{W}_\epsilon$ is continuous. Now for any traceless 2-tensor $\Sigma \in C^0$ satisfying $|\Sigma|^2 \rightarrow \mr\sigma^2$ on the ends at the rate $e^{-\delta t}$, we define $\mathcal{Q}(\Sigma)$ to be the unique solution of the Lichnerowicz equation (\ref{origLich}) with $\sigma +LW$ replaced by $\Sigma$ which satisfies $\phi_- \leq \mathcal{Q}(\Sigma) \leq \phi_+$. This map is shown to be well-defined in \cite{Leach14}. We thus define a map $\mathcal{S}_{\sigma} : C^0 _{-\delta} (S^2 _0 (M)) \rightarrow W^{2,p} _{-\nu}(M)$ by
\begin{equation}
\pi \mapsto \mathcal{Q}(\sigma + \pi) - \mr\phi.
\end{equation}

\noindent If we could show this map to be continuous, then the map $\mathcal{N}_{\sigma,\epsilon} = \mathcal{S}_\sigma \circ L \circ \mathcal{W}_\epsilon$ would thus be continuous itself. The range of this function lies in $U$ by definition, and the composition of $\mathcal{N}_{\sigma,\epsilon}$ with the compact embedding $W^{2,p} _{-\nu} \hookrightarrow C^0 _{-\nu'}$ gives us a map $T$ which satisfies the hypothesis of Theorem \ref{Schauder}. If $\tilde\psi$ is this fixed point, then $(\mr\phi + \tilde\psi, W_{\mr\phi + \tilde\psi})$ is a solution of the subcritical system by construction. We thus need only show that $\mathcal{S}_\sigma$ is continuous, and to do this we require the following lemma.
\begin{lem}\label{continuity}
The map $\mathcal{S}_{\sigma} : C^0 _{-\delta} (S^2 _0 (M)) \rightarrow W^{2,p} _{-\nu} (M)$ is continuous.
\end{lem}

\noindent The proof of this lemma is an implicit function theorem argument which goes through exactly as the proof of \cite[Lem 4.2]{Leach14} with the obvious modifications.


\section{Convergence of Solutions}

In this section we show that any solution of the subcritical equations has an $L^\infty$ bound depending only on the $L^2$-norm of $LW$. We essentially follow the proof of this for closed manifolds found in \cite{DGH11}, though the geometry of the ends clearly necessitates several modifications to their argument. Let $\epsilon \in [0,1)$ be arbitrary, and let $(\phi,W)$ be a solution to the subcritical equations \eqref{subcrit1}-\eqref{subcrit2}. We define the energy of the solution by
\[
\gamma(\phi,W) = \int_M |LW|^2\, dv,
\] and let $\tilde\gamma = \max\{\gamma,1\}$. Note that $\tilde\gamma$ is finite by Lemma \ref{dwbd}. We rescale $\phi$, $W$ and $\sigma$ as
\[
\tilde\phi = \tilde\gamma^{-\frac{1}{2N}} \phi, \,\,\,\, \tilde W = \tilde\gamma^{-\frac{1}{2}} W, \,\,\,\, \tilde\sigma = \tilde\gamma^{-\frac{1}{2}} \sigma.
\] The deformed equations can then be renormalized as
\begin{equation}\label{newLich}
\frac{1}{\tilde\gamma^{1/n}}(\Delta \tilde\phi - c_n R \tilde\phi) = b_n \tau^2 \tilde\phi^{N-1} - c_n |\tilde\sigma + L\tilde W|^2 \tilde\phi^{-N-1},
\end{equation}\begin{equation}\label{newVect}
\di L\tilde W = \frac{n-1}{n} \tilde\gamma^{-\frac{\epsilon}{2N}}\tilde\phi^{N-\epsilon} d\tau.
\end{equation} Notice that because of our rescaling, we have
\[
\int_M |L\tilde W|^2\, dv \leq 1.
\]

Throughout this section, ``bounded" will mean ``bounded independent of $\epsilon, \phi$ and $W$", and all constants $C$ or $C_i$ will be similarly independent of $\epsilon$, $\phi$ and $W$. We first prove an important lemma.

\begin{lem}\label{bound}
  Suppose that $k\geq 0$. Then, for any solution $\tilde\phi$ of the renormalized subcritical equations \eqref{newLich}-\eqref{newVect}, and any $\delta>0$, we have
\begin{multline}\label{boundEqn}
-C_1\left(\int_{M} e^{-\delta t} \tilde \phi^{2N+Nk} \right)^{\frac{N+2+Nk}{2N+Nk}} + b_n \tau_0^2 \int_M e^{-\delta t}\tilde\phi^{2N+Nk} \\ \leq 2 c_n\int_M e^{-\delta t}|\sigma|^2 \tilde\phi^{Nk} + C_2\int_M |L\tilde W|^2 \tilde\phi^{Nk}.
\end{multline}
\end{lem}
\begin{proof}
We multiply equation \eqref{newLich} by $e^{-\delta t} \tilde\phi^{N+1+Nk}$ and integrate over $M$ to get
\begin{multline}\label{integratedLich}
\frac{1}{\tilde\gamma^{1/n}}\int_M \left(-e^{-\delta t}\tilde\phi^{N+1+Nk}\Delta \tilde\phi + c_n e^{-\delta t}R \tilde\phi^{N+2+Nk}\right) dv \\+ b_n \int_M \tau^2 e^{-\delta t} \tilde\phi^{2N+Nk} dv = c_n\int_M e^{-\delta t}|\tilde\sigma + L \tilde W|^2 \tilde\phi^{Nk}dv.
\end{multline}

Consider the first integral on the left. We have
\begin{align*}
  \int_M - e^{-\delta t} \tilde \phi^{N+1+Nk} \Delta \tilde\phi &= \int_M c_1 d (e^{-\delta t}) d(\tilde\phi^{N+2+Nk})  + \int c_2 e^{-\delta t} \tilde\phi^{N+Nk} |d\tilde\phi|^2 \\
  &\geq - C\int_M \Delta(e^{-\delta t}) \tilde\phi^{N+2+Nk} \\
  &\geq - C\int_M e^{-\delta t}\tilde\phi^{N+2+Nk}
\end{align*} where $c_1 = \frac{1}{N+2+Nk}$ and $c_2 = N+1+Nk$. The first and second lines are by integration by parts. These integration by parts are valid because  of the exponential falloff term $e^{-\delta t}$. The third line follows from the inequality $\Delta e^{-\delta t} \leq C e^{-\delta t}$ since $\|t\|_{C^2} <\infty$. Using this, and combining with the $R$ term, the first integral in equation \eqref{integratedLich} is greater than or equal to
\begin{align*}
\int_M (c_n R-C) e^{-\delta t} \tilde\phi^{N+2+Nk}.
\end{align*} Let $v = \frac{2N + Nk}{N+2 + Nk}$, $u = \frac{2N +Nk}{N-2}$. Note that $\frac{1}{u} + \frac{1}{v} = 1$. Let $\mu = 1/u$. We then use H\"older's inequality to see that
\begin{align*}
\int_M (c_n R-C) e^{-\delta t} \tilde\phi^{N+2+Nk} &= \int_M (c_n R-C) e^{-\delta t \mu} \tilde\phi^{N+2+Nk} e^{-\delta t(1-\mu)}\\
& \geq -\||c_n R-C| e^{-\delta t\mu}\|_{L^u} \|\tilde\phi^{N+2+Nk} e^{-\delta t(1-\mu)}\|_{L^v} \\
&\geq - \|c_n R-C\|_{u, \delta/u} \left(\int_M e^{-\delta t} \tilde\phi^{2N+Nk} \right)^{\frac{N+2+Nk}{2N+Nk}}.
\end{align*} Note that $\|c_n R-C\|_{u,\delta/u} <\infty$ since $R\in L^\infty$ and $\delta/u>0$. After using $\tilde\gamma^{-1} \leq 1$, this gives us the first term of the desired inequality.

The second term is easily found by pulling out the infimum of $\tau^2$. The right hand term is found by the inequalities $|\tilde\sigma + L\tilde W|^2 \leq 2|\tilde\sigma|^2 + 2|L\tilde W|^2$, $e^{-\delta t} \leq C$ and $|\tilde\sigma|^2 \leq |\sigma|^2$. This completes the lemma.
\end{proof}

\begin{prop}\label{phiBound}
Suppose $\phi$ is a positive solution of the subcritical equations (\ref{subcrit1})-(\ref{subcrit2}) for tame initial data which satisfies $\phi \to \mr\phi$ on the ends. If $\epsilon\in [0,1)$, we have
\[
\phi < C \tilde\gamma^{\frac{1}{2N}}.
\]
\end{prop}
\begin{proof}

As in \cite{DGH11}, we will prove this proposition in four steps.

\textbf{Step 1.} \emph{$L^1_\delta$ bound on $\tilde\phi^{2N}$}

Suppose $d\tau\in L^p_{-\delta}$. Then using Lemma \ref{bound} with $k = 0$, $\tilde \phi^{2N}$ is clearly bounded in $L^1_\delta$ as long as the right hand side of \eqref{boundEqn} is bounded. However, since $k = 0$ and $\sigma \to \mr{\sigma}$, it is easy to see that both integrals are bounded. Thus $\tilde\phi^{2N}$ is bounded in $L^1_\delta$.

\textbf{Step 2.} \emph{Bounds for $LW$.}

Suppose by induction we have $\tilde\phi^{p_iN}$ bounded in $L^1_\delta$ for some $2\leq p_i$. Let $\frac{1}{q_i} = \frac{1}{p_i} + \frac{1}{p}$. If $q_i>n$, we continue on to step 4. Otherwise define $\frac{1}{r_i} = \frac{1}{q_i} - \frac{1}{n}$.  We will show at the end of step 3 that we can ensure $q_i$ is never $n$.

Let $\alpha = \max\{\delta\left(\frac{1}{p_i} - 1\right), -\delta_*/2\}$. Note that
\[
\tilde \phi^{N-\epsilon} \leq \frac{N-\epsilon}N \tilde\phi^N + \frac\epsilon{N} \leq \tilde \phi^N + \frac1N.
\] Using this and Equation \eqref{newVect}, we get
\begin{align*}
\|\di L\tilde W\|_{L^{q_i}_\alpha} &= C \tilde\gamma^{-\frac{\epsilon}{2N}} \| \tilde \phi^{N-\epsilon} |d\tau|\|_{L^{q_i}_\alpha}\\
& \leq C \left\|\left(\tilde\phi^N+\frac1N\right) |d\tau|\right\|_{L^{q_i}_{\delta\left(\frac{1}{p_i} - 1\right)}} \\
&\leq C \left\|\left(\tilde\phi^N+\frac1N\right)|d\tau| e^{-\frac{\delta}{p_i} t} e^{ \delta t}\right\|_{L^{q_i}} \\
&\leq C \left\|\left(\tilde\phi^N+\frac1N\right) e^{-\frac{\delta}{p_i} t} \right\|_{L^{p_i}} \||d\tau| e^{\delta t}\|_{L^{p}}\\
&\leq C\left(\|\tilde\phi^{Np_i} \|^{1/p_i}_{L^1_\delta} + \|1/N\|_{L^{p_i}_{\delta/p_i}}\right) \|d\tau\|_{L^{p}_{-\delta}}.
\end{align*} The second line holds because $\tilde\gamma\geq 1$. The fourth line is H\"older's inequality with $p_i$ and $p$. The last line follows from the definitions of the norms and the triangle inequality.

This last inequality shows that $\|\di L\tilde W\|_{L^{q_i}_\alpha}$ is bounded since the first norm on the right is bounded by hypothesis and the second is bounded since $d\tau \in L^p_{-\delta}$. We also calculate, since $q_i<n$,
\begin{align*}
  \|L\tilde W\|_{L^{r_i}_{\delta/r_i}} &\leq C\|L\tilde W\|_{W^{1,q_i}_{\alpha}} \\
  &\leq C\|\tilde W\|_{W^{2,q_i}_\alpha \oplus \mathscr{Y}} \\
  &\leq C\|\di L\tilde W\|_{L^{q_i}_\alpha}.
\end{align*} The first line follows from Theorem \ref{embeddings} and the fact that $\delta/r_i > \alpha$. The second line is because $L$ maps $\mathscr{Y}$ to an exponentially decaying piece. The last line is by the existence of generalized inverse of $\di L$ implied by Theorem \ref{conformalVectorLaplacian}.

We have thus shown that $\|L\tilde W\|_{L_{\delta/r_i}^{r_i}}$ is bounded.

\textbf{Step 3.} \emph{Induction on $p_i$}

Define $k_i$ by $\frac{2}{r_i} + \frac{k_i}{p_i} =1$. Lemma \ref{bound} implies that we can show that $\tilde\phi^{2N+Nk_i}$ is bounded in $L^1_\delta$ as long as
\[
2c_n \int_M e^{-\delta t}|\sigma|^2 \tilde\phi^{Nk_i} + C_2\int_M |L\tilde W|^2 \tilde\phi^{Nk_i}
\] is bounded. For both integrals, we use H\"older's inequality with $\frac{r_i}{2}$ and $\frac{p_i}{k_i}$ to bound this above by
\[
C\|\tilde\phi^{p_i N} \|_{L^{1}_\delta}^{k_i/p_i} \left(\|\sigma\|_{L^{r_i}_{\delta/r_i}}^2 + \|L\tilde W\|_{L^{r_i}_{\delta/r_i}}^2\right).
\] The norm on $\tilde \phi$ is bounded by assumption. The norm on $\sigma$ is bounded since $\delta/r_i>0$ and by our conditions on $\sigma$. The norm on $L\tilde W$ is bounded by the previous step.

All of this shows that $\tilde \phi^{2N+Nk_i}$ is bounded in $L^1_\delta$. Let $p_{i+1} = 2+k_i$. With this we have
\begin{multline*}
\frac{p_{i+1}}{p_i} = \frac{2+k_i}{p_i} = \frac{2}{p_i} +1 - \frac{2}{r_i} = \frac{2}{p_i} +1 -2\left(\frac{1}{q_i} - \frac{1}{n}\right) \\= \frac{2}{p_i} + 1 -2\left(\frac{1}{p_i} + \frac{1}{p} - \frac{1}{n}\right)  = 1+ \frac{2}{n}-\frac{2}{p} > 1
\end{multline*} since $p>n$. Hence $p_i\to \infty$, and so $q_i \to p$. We continue steps 2 and 3 a finite number of times until some $k$ such that $q_k >n$. We can avoid the case that $q_i =n$ by slightly decreasing $p$ and $\delta$ at the beginning of the proposition, since $L^p_{-\delta} \subset L^{p-\epsilon}_{-\delta + \epsilon}$ for small $\epsilon>0$.

\textbf{Step 4.} \emph{$L^\infty$ bound on $\tilde \phi$.}

Since $q_k>n$, we have, similar to step 2,
\begin{align*}
\|L\tilde W\|_{L^\infty_{\alpha}} \leq C\|\tilde W\|_{W^{2,q_k}_{\alpha}} &\leq C\| \di L\tilde W\|_{L^{q_k}_{\alpha}}\\ &\leq C\left(\|\tilde\phi^{Np_i} \|^{1/p_i}_{L^1_\delta} + \|1/N\|_{L^{p_i}_{\delta/p_i}}\right) \|d\tau\|_{L^{p}_{-\delta}}
\end{align*} where the right hand side is again bounded. Since $\alpha<0$, this implies that $|L\tilde W|$ is bounded as well.

From the fact that the Laplacian acting on functions only involves first order derivatives of the metric, and since the coefficients of the Lichnerowicz equation \eqref{newLich} are at least in $C^{0,\beta}$ for some $\beta>0$ since $p>n$  it can be easily seen that the function $\tilde \phi$ is in $C^{2,\beta}$. We can thus apply the maximum principle. Let $x\in M$ be where $\tilde \phi$ reaches its maximum value, if it has one. At such a point, we have
\[
\frac{c_n}{\tilde \gamma^{1/n}} R \tilde\phi +  b_n \tau^2 \tilde\phi^{N-1} \leq c_n |\tilde\sigma + L\tilde W|^2 \tilde\phi^{-N-1}
\] which simplifies to
\[
\frac{c_n}{\tilde \gamma^{1/n}} R \tilde\phi^{N+2} + b_n \tau^2 \tilde\phi^{2N} \leq c_n |\tilde\sigma + L\tilde W|^2.
\] Since $R\in L^\infty$ and $\tilde\gamma\geq 1$, $\tilde\phi$ is bounded.

If $\phi>\sup_\Sigma \mr{\phi}$ at some point, it (and thus $\tilde\phi$) has a maximum. Thus, if $\phi$ does not have a maximum, $\phi \leq \sup_\Sigma \mr{\phi}$, which is an even stronger upper bound than the proposition requires.

By recalling that $\tilde\phi = \tilde\gamma^{-\frac{1}{2N}} \phi$, we have proven the proposition.
\end{proof}

Now that we have the bound, let us consider what happens as $\epsilon \to 0$.

\begin{lem}\label{convergenceToSolution}
  Assume that there exist sequences $\epsilon_i$ and $(\phi_i, W_i)$ such that $\epsilon_i \geq 0$, $\epsilon_i \to 0$ and $(\phi_i, W_i)$ is a solution of the deformed equations \eqref{subcrit1}-\eqref{subcrit2} with $\epsilon = \epsilon_i$. Also assume that $\gamma(\phi_i, W_i)$ is bounded. Then there exists a constant $\nu>0$ and a sequence of the $(\phi_i, W_i)$ which converges in the $W^{2,p}_{-delta}\oplus\mathscr{Y}$ norm to a solution $(\phi_\infty,W_\infty)$ of the original conformal constraint equations.
\end{lem}
\begin{proof}
From the previous proposition, we know that the $\phi_i$ are uniformly bounded in the $L^\infty$ norm. By the now standard inequality \begin{equation}\label{easyWBound}\|W_i\|_{W^{2,p}_{-\delta}\oplus\mathscr{Y}} \leq C\left(\|\phi_i \|_0^N + \|1/N\|_0\right) \|d\tau\|_{L^{p}_{-\delta}},\end{equation} the sequence $W_i$ is uniformly bounded in $W^{2,p}_{-\delta} \oplus\mathscr{Y}$. We have that $L: W^{2,p}_{-\delta}\oplus\mathscr{Y} \to C^0_{-\delta'}$ is compact by Theorem \ref{embeddings} and the fact that $\mathscr{Y}$ is finite dimensional. Thus, up to selecting a subsequence, we can assume that the sequence $LW_i$ converges in $C^0_{-\delta'}$ to some $LW_\infty$.

Thus by Lemma \ref{continuity}, the functions $\psi_i := \phi_i - \mr\phi$ converge in $W^{2,p}_{-\nu}$ (and thus in $L^\infty_{-\nu}$) for some $\nu>0$ to a function $\psi_\infty$. We must assume we picked $\delta'$ close enough to $\delta$ in the previous paragraph such that we have $\nu < \frac{\delta'}{2N}$. Since $\phi_i \rightarrow \mr\phi$ on the ends, this implies that $\phi_\infty := \psi_\infty + \mr\phi$ approaches $\mr\phi$ on the ends exponentially fast. Since the right hand side of the vector equation \eqref{origVect} converges in $L^{p}_{-\delta}$, we have that the sequence $W_i$ converges in the $W^{2,p}_{-\delta}\oplus\mathscr{Y}$ norm as well. The regularity of $\phi_\infty$ and $W_\infty$ guarantee that they are solutions of the conformal constraint equations (with $\epsilon =0$).
\end{proof}

\begin{lem}\label{convergenceToLimitEquation}
  Assume there exists sequences $\epsilon_i$ and $(\phi_i, W_i)$ such that $\epsilon_i\geq 0$, $\epsilon_i\to 0$ and $(\phi_i, W_i)$ is a solution of the subcritical equations \eqref{subcrit1}-\eqref{subcrit2} with $\epsilon = \epsilon_i$. Also assume that $\gamma(\phi_i, W_i) \to \infty$. Then there exists a non-zero solution $W \in W^{2,p}_{-\delta}\oplus \mathscr{Y}$ of the limit equation
  \[
  \di LW = \alpha_0 \sqrt{\frac{n-1}{n}} |LW| \frac{d\tau}{\tau}
  \] for some $\alpha_0 \in (0,1]$ such that $|LW| \leq C e^{-\delta t}$ for some $C$ independent of $\phi_i$, $W_i$ and $W$.
\end{lem}
\begin{proof}
Arguing as in the previous lemma, we have that $\tilde W_i$ are uniformly bounded in $W^{2,p}_{-\delta}\oplus \mathscr{Y}$. Without loss of generality, we can assume that $\gamma>1$, and so
\[
\int_M |L \tilde W_i|^2 = 1.
\] Up to selecting a subsequence, we can then assume that $L\tilde W_i$ converges in $C^{0}_{-\delta'}$ for some $\alpha>0$ to some $L\tilde W_\infty$.

We can show the falloff of $L\tilde W_\infty$ by considering
\begin{align*}
\|LW_i\|_{L^\infty_{-\delta}} &\leq C\|\phi_i^{N-\epsilon_i} |d\tau|\|_{L^p_{-\delta}}\\
&\leq C\left\|(\phi_i^N+1/N) |d\tau|\right\|_{L^{p}_{-\delta}} \\
&\leq C \tilde\gamma^{1/2}_i \|d\tau\|_{L^p_{-\delta}}
\end{align*} as before, but using Proposition \ref{phiBound} and $\tilde\gamma\geq 1$. Thus
\[
\|L\tilde W_i\|_{L^\infty_{-\delta}}\leq C
\] for some $C$ independent of $\epsilon$, $W_i$ and $\phi_i$. Since the convergence of $L\tilde W_i$ is in $L^\infty_{-\delta'}$, we have
\[
|L\tilde W_\infty| \leq C e^{-\delta' t}
\] for some $C$ independent of $\phi_i$, $W_i$ and $W$. After the rest of this proof, we can repeat this argument with better convergence to get the desired fall off.

Let $\tilde \phi_\infty$ be defined by
\[
\tilde \phi_\infty^N = \sqrt{\frac{n}{n-1}} \tau^{-1} |L\tilde W_\infty|
\] i.e., $\tilde\phi_\infty$ satisfies
\[
b_n \tau^2 \tilde\phi_\infty^{N-1} = c_n|L\tilde W_\infty|^2 \tilde\phi_\infty^{-N-1}.
\] If we can show that $\tilde\phi_i \to \tilde \phi_\infty$ in $L^\infty$, then the continuity of the vector equation implies we have that $\tilde W_\infty$ is a $W^{2,p}_{-\delta}\oplus \mathscr{Y}$ solution to the limit equation with $\alpha_0 = \lim \gamma(\phi_i, W_i)^{-\frac{\epsilon_i}{2N}}$. We have $\alpha_0\in [0,1]$ since $\gamma(\phi_i, W_i) \to \infty$. Note that \[\int_M |L\tilde W_\infty|^2 = 1\] since $L\tilde W_i$ converges in $C^0_{-\delta'}$, and so $\tilde W_\infty \not\equiv 0$, and so the solution is nontrivial. Since we assumed there are no global conformal Killing fields in $L^2$, we cannot have the case $\alpha_0 = 0$.

To show this convergence, we will show that for any $\epsilon >0$, that $|\tilde\phi_\infty - \tilde\phi_i|<\epsilon$ for large enough $i$. Take a $C^2$ function with bounded derivatives $\tilde\phi_+$ such that
\[
\tilde\phi_\infty + \frac{\epsilon}{2} \leq \tilde\phi_+ \leq \tilde\phi_\infty + \epsilon.
\] We show that $\tilde\phi_+$ is a supersolution of the rescaled Lichnerowicz equation \eqref{newLich} if $i$ is large enough. Multiplying the rescaled Lichnerowicz equation \eqref{newLich} by $\tilde\phi_+^{N+1}$, we have to show that
\[
\frac{\tilde\phi_+^{N+1}}{\tilde\gamma^{1/n}} \left(-\Delta \tilde\phi_+ + c_n R \tilde\phi_+\right) + b_n \tau^2 \tilde\phi_+^{2N} \geq c_n |\tilde \sigma + L\tilde W_i|^2.
\] Since
\[
\tilde\phi_+^{2N} \geq \left(\tilde\phi_\infty + \frac{\epsilon}{2}\right)^{2N} \geq \tilde\phi_\infty^{2N} + \left(\frac{\epsilon}{2}\right)^{2N},
\] the previous inequality will be satisfied provided that
\[
\frac{\tilde\phi_+^{N+1}}{\tilde\gamma^{1/n}} \left(- \Delta \tilde\phi_+ + c_nR \tilde\phi_+\right) + b_n \tau^2\left(\frac{\epsilon}{2}\right)^{2N} \geq c_n |\tilde \sigma + L\tilde W_i|^2 - c_n |L\tilde W_\infty|^2.
\] Note that everything goes to zero as $i\to \infty$ except for the $\epsilon$ term. Since $\tau^2 \geq \tau_0^2$, there exists an $i_0$ such that for all $i\geq i_0$, $\tilde\phi_+$ is a supersolution.

Note that $\tilde\phi_+ \geq \frac{\epsilon}{2}$. Also, $\tilde\phi_i \to \tilde{\mr\phi} = \tilde\gamma_i^{-1/2N}\mr\phi$ on the ends for every $i$. Thus, for $i$ large enough, $\tilde\phi_i <\frac{\epsilon}{2}$ outside some compact set $K_i$. Inside $K_i$, since $\tilde\phi_\infty$ is a supersolution and $\tilde\phi_i$ is regular enough, we can apply the maximum principle to show that $\tilde\phi_+$ remains larger than $\tilde \phi_i$. Thus $\tilde\phi_i \leq \tilde\phi_+ \leq \tilde\phi_\infty +\epsilon$ for large enough $i$.

We proceed similarly with a $\tilde\phi_- \in C^2$ with
\[
\tilde\phi_\infty-\epsilon \leq \tilde\phi_- \leq \tilde\phi_\infty -\frac{\epsilon}{2}.
\] Since $LW_\infty\to 0$ on the ends, $\tilde\phi_-$ is negative on the ends. On the set where it is positive, however, we can show that it is also a subsolution to the rescaled Lichnerowicz equation \eqref{newLich}. By the same argument as before, $\tilde\phi_i \geq \tilde\phi_- \geq \tilde\phi_\infty -\epsilon$. This completes the theorem.

\end{proof}

We can now prove our main result.

\begin{thm}
  Let $(g,\sigma,\tau)$ be conformal data on an asymptotically cylindrical manifold satisfying the conditions of Theorem \ref{subcritthm}. Then at least one of the following is true:
  \begin{itemize}
    \item The system \eqref{origLich}-\eqref{origVect} admits a solution $(\phi,W)$ with regularity as in Theorem \ref{subcritthm}. Also, the set of these solutions is compact.

    \item There exists a non-zero solution $W \in W^{2,p}_{-\delta}\oplus \mathscr{Y}$ of the limit equation
  \[
  \di LW = \alpha_0 \sqrt{\frac{n-1}{n}} |LW| \frac{d\tau}{\tau}
  \] for some $\alpha_0 \in (0,1]$ such that $|LW| \leq C e^{-\delta t}$ for some $C$ independent of $\phi_i$, $W_i$ and $W$.
  \end{itemize}
\end{thm}
\begin{proof}
  Assume that the limit equation admits no such solution for any $\alpha_0\in (0,1]$. From Theorem \ref{subcritthm}, we know there exists a sequence of solutions $(\phi_i, W_i)$ with appropriate regularity of the deformed constraints \eqref{subcrit1}-\eqref{subcrit2} with $\epsilon_i = 1/i$. If the sequence $\gamma(\phi_i, W_i)$ was unbounded, there would be a non-zero solution to the limit equation by Lemma \ref{convergenceToLimitEquation}, a contradiction. Thus the sequence is bounded, and so by Lemma \ref{convergenceToSolution} there exists a solution $(\phi_\infty, W_\infty)$ with appropriate regularity of the conformal constraint equations \eqref{origLich}-\eqref{origVect}.

  For compactness, let $(\phi_i, W_i)$ be an arbitrary sequence of solutions to the conformal constraint equations. Using Lemma \ref{convergenceToLimitEquation} with $\epsilon_i = 0$, we have that $\gamma(\phi_i, W_i)$ is bounded. Lemma \ref{convergenceToSolution} then says that a subsequence of $(\phi_i, W_i)$ converges. This completes the proof.
\end{proof}


\section{Existence Results}

Theorem \ref{mainTheorem} says that if we can show that the limit equation \eqref{limiteq} has no solutions with particular properties, then there is a solution to the full constraint equations (\ref{origLich})-(\ref{origVect}). In this section, we will use this result to show that for certain tame near-CMC seed data, there is no solution to the limit equation. Hence Theorem \ref{mainTheorem} guarantees a solution to the constraint equations.

\begin{cor}
Let $(g,\sigma,\tau)$ be conformal data on an asymptotically cylindrical manifold satisfying the conditions of Theorem \ref{subcritthm}, and suppose $\Ric \leq (c_1\chi^2 - c_2 e^{-2\mu t}) g$ for some constants $c_i, \mu >0$ and smooth compactly supported bump function $\chi$. Then there is some $C > 0$ such that if
\[
\left\|\frac{d\tau}{\tau} \right\|_{C^0_{-\mu}} < C,
\] there is a solution to the LCBY equations (\ref{origLich})-(\ref{origVect}) as in Theorem \ref{mainTheorem}.
\end{cor}
\begin{proof}
Assume $W \in W^{2,p}_{-\delta}\oplus \mathscr{Y}$ is a solution of the limit equation. We claim that for any such $W$, \[
\int |LW|^2 \geq C \int |W|^2 e^{-2\mu t}
\] for some positive $\mu \leq \delta, C >0$. If so, we can then take the limit equation, multiply by $W$ and integrate by parts to get
\begin{align*}
\int |LW|^2 &\leq -\int \sqrt{\frac{n-1}{n}} |LW| W \frac{d\tau}{\tau}\\
&\leq \left\|\frac{d\tau}{\tau} \right\|_{C^0_{-\mu}} \|LW\|_{L^2} \left( \int |W|^2 e^{-2\mu t} \right)^{1/2}.
\end{align*} We then immediately see that there is a constant $C$ (the one used in the hypotheses) such that
\[
\left\|\frac{d\tau}{\tau} \right\|_{C^0_{-\mu}} \geq C.
\] This contradicts our hypotheses, and so there is a solution to the constraint equations.

We prove the claimed inequality first for compactly supported vector fields $V\in W^{2,p}$. First, recall the pointwise Bochner type formula
\[
\frac12 \di L V = \Delta V + \left(1-\frac2n\right) \nabla(\di V) + \Ric(V,\cdot),
\] which is shown, for example, in \cite[App B]{GS12}. Multiplying both sides by $V$ and integrating by parts, we get
\begin{align*}
\frac12 \int |LV|^2 = \int |\nabla V|^2 + \left(1-\frac2n\right) (\di V)^2 - \Ric(V,V).
\end{align*} Dropping the $\di V$ term and using the Ricci curvature bound, we get
\begin{align} \label{eq:endInequality}
\frac12 \int |LV|^2 + c_1 \int \chi^2 |V|^2 &\geq \int |\nabla V|^2 + c_2 \int |V|^2 e^{-2\mu t} \\
&\geq C \|V\|_{W^{1,2}(\supp\chi)}^2 \notag
\end{align} where the $C$ depends on $\chi$ and $\mu$.

Next, we want to show that $\int |LV|^2 \geq C \int \chi^2 |V|^2$. Assume this were not true. Then there exists a sequence $V_i$ such that
\[
\int |LV_i|^2 \leq \frac1i \int \chi^2 |V_i|^2.
\] We normalize the $V_i$ such that $\int \chi^2 |V_i|^2 = 1$. Because of inequality \eqref{eq:endInequality}, we have that $\|V_i\|_{W^{1,2}(\supp\chi)}$ is bounded. Thus $V_i$ converge strongly to some $V$ in $L^2(\supp\chi)$ by the Rellich-Kondrachov theorem. In particular, $\int \chi^2 |V|^2 =1$ and so it is nonzero. Also, since $\int |\nabla V|^2$ is bounded, we get weak convergence of $|LV|$, and so $\int |LV|^2 = 0$. This implies $V$ is a nontrivial global $L^2$ conformal Killing field, contradicting our assumptions. Thus $\int |LV|^2 \geq C \int \chi^2 |V|^2$.

This immediately gives that\[
\int |LV|^2 \geq C \int |V|^2 e^{-2\mu t}
\] for compactly supported $V$. We claim the same inequality holds for $W\in W^{2,p}_\delta\oplus \mathscr{Y}$. Indeed, there are smooth cutoff functions $\eta_i$ such that for $W_i = \eta_i W$,
\[
\left| \int |LW|^2 - \int |LW_i|^2 \right| < \frac1i \,\,\,\, \textrm{ and } \,\,\,\, \left|\int |W|^2 e^{-2\mu t} - \int |W_i|^2 e^{-2\mu t} \right| < \frac1i.
\] This is because $LW$ decays exponentially fast outside a compact set and because $|W|$ is bounded. Thus we have
\begin{align*}
C\int |W|^2 e^{-2 \mu t} &\leq C\int |W_i|^2 e^{-2\mu t} + \frac{C}i \\
&\leq (1+\epsilon) \int |LW_i|^2 - C\epsilon \int |W_i|^2 e^{-2\mu t} + \frac{C}i\\
&\leq (1+\epsilon) \int |LW|^2 - C\epsilon \int |W_i|^2 e^{-2\mu t} + \frac{C +1+\epsilon}{i}
\end{align*} for some small fixed $\epsilon>0$. Thus, for large enough $i$, the last two terms add together to be negative, and so we have
\[
C\int |W|^2 e^{-2 \mu t} \leq (1+\epsilon) \int |LW|^2
\] for any small enough $\epsilon>0$. Thus the desired inequality holds. This completes the proof.
\end{proof}


\section{Extension of main results to conformally asymptotically cylindrical metrics}

While the results in this paper have been proven only for asymptotically cylindrical manifolds, analogous results hold for conformally asymptotically cylindrical manifolds (see Subsection \ref{notation}) as well with only a few changes in the proof. First, as explained in \cite{Leach14}, the $L^p$-Sobolev version of Theorem \ref{conformalVectorLaplacian} holds even for conformally AC metrics, and this observation is the basis for the proof of Lemma \ref{dwbd}. Now let $w$ be any conformal factor (as in Subsection 1.1) such that $\check{g} = w^{N-2} g$, where $g$ is an asymptotically cylindrical manifold. It follows from the covariance of the Lichnerowicz equation that $\phi$ is a global sub/supersolution of the conformal LCBY equations \eqref{origLich}-\eqref{origVect} for $\check{g}$ if and only if $w\phi$ is a global sub/supersolution of the equation
\begin{equation}\label{conformalLich}
\Delta_g \theta - c_n R_g \theta = b_n \tau^2 \theta^{N-1} - c_n |w^2 (\sigma + L_{\check{g}} W)|^2 _g\theta^{-N-1}
\end{equation} coupled with the vector equation \eqref{origVect} for $\check{g}$. Notice that only the $L$ operator is defined with respect to $\check{g}$, while the rest are with respect to $g$.

Because of this fact, we can find global sub/supersolutions $\phi_-, \phi_+$ to Equation \eqref{conformalLich} as in Section \ref{barriers}, since that metric is asymptotically cylindrical. Then, $\phi_-/w, \phi_+/w$ are global sub and supersolutions to the original conformal constraint equations for $\check{g}$. Note that these will produce a solution $\phi$ to the Lichnerowicz equation that asymptotes to $\mr\phi/\mr w_\ell$, which is still a valid asymptote since $\mr w_\ell$ is also a function on $N_\ell$. The rest of the proof proceeds the same.


\section{Acknowledgments}
The first author was partially supported by the NSF grant DMS-1263431. This material is based upon work supported by the National Science Foundation under Grant No. 0932078 000, while the first author was in residence at the Mathematical Sciences Research Institute in Berkeley, California, during the fall of 2013. The second author was partially supported by the NSF grant DMS-1105050.

The authors would like to thank Jim Isenberg and Rafe Mazzeo for suggesting this project.

\bibliographystyle{alpha}
\bibliography{KnownResults}

\begin{thebibliography}{DIMM13}

\bibitem[ACI08]{ACI08}
P.~Allen, A.~Clausen, and J.~Isenberg.
\newblock Near-constant mean curvature solutions of the {E}instein constraint
  equations with non-negative {Y}amabe metrics.
\newblock {\em Class. Quantum Grav.}, 25(7):075009, 15, 2008.

\bibitem[CM]{CM12}
P.~Chru\'{s}ciel and R.~Mazzeo.
\newblock Initial data sets with ends of cylindrical type: {I}. {T}he
  {L}ichnerowicz equation.
\newblock {\em To appear Monatsh. Math.}

\bibitem[CMP]{CMP12}
P.~Chru\'{s}ciel, R.~Mazzeo, and S.~Pocchiola.
\newblock Initial data sets with ends of cylindrical type: {II}. {T}he vector
  constraint equation.
\newblock {\em To appear Adv. Theor. Math. Phys.}

\bibitem[DGH11]{DGH11}
M.~Dahl, R.~Gicquaud, and E.~Humbert.
\newblock A limit equation associated to the solvability of the vacuum
  {E}instein constraint equations using the conformal method.
\newblock {\em Unpublished}, 2011.
\newblock arXiv:1012.2188.

\bibitem[DGI14]{DGI14}
J.~Dilts, R.~Gicquaud, and J.~Isenberg.
\newblock A limit equation criterion for applying the conformal method to
  asymptotically {E}uclidean initial data sets.
\newblock {\em Preprint}, 2014.

\bibitem[Dil13]{Dilts13}
J.~Dilts.
\newblock The {E}instein constraint equations on compact manifolds with
  boundary.
\newblock 2013.
\newblock arXiv:1310.2303.

\bibitem[DIMM13]{DIMM13}
J.~Dilts, J.~Isenberg, R.~Mazzeo, and C.~Meier.
\newblock Non-{CMC} solutions of the {E}instein constraint equations on
  asymptotically {E}uclidean manifolds.
\newblock 2013.
\newblock arXiv:1312.0535.

\bibitem[Gou12]{G12}
E.~Gourgoulhon.
\newblock {\em 3+1 Formalism and Bases in Numerical Relativity}.
\newblock Springer-Verlag, 2012.

\bibitem[GS12]{GS12}
R.~Gicquaud and A.~Sakovich.
\newblock A large class of non-constant mean curvature solutions of the
  {E}instein constraint equations on an asymptotically hyperbolic manifold.
\newblock {\em Commun. Math. Phys}, 310:705--763, 2012.

\bibitem[HMT13]{HMT13}
M.~Holst, C.~Meier, and G.~Tsogtgerel.
\newblock Non-cmc solutions of the {E}instein constraint equations on compact
  manifolds with apparent horizon boundaries.
\newblock 2013.
\newblock arXiv:1310.2302.

\bibitem[HNT09]{HNT09}
M.~Holst, G.~Nagy, and G.~Tsogtgerel.
\newblock Rough solutions of the {E}instein constraints on closed manifolds
  without near-{CMC} conditions.
\newblock {\em Comm. Math. Phys.}, 288(2):547--613, 2009.

\bibitem[IM96]{IM96}
J.~Isenberg and V.~Moncrief.
\newblock A set of nonconstant mean curvature solutions of the {E}instein
  constraint equations on closed manifolds.
\newblock {\em Class. Quantum Grav.}, 13(7):1819--1847, 1996.

\bibitem[Ise95]{Isenberg95}
J.~Isenberg.
\newblock Constant mean curvature solutions of the {E}instein constraint
  equations on closed manifolds.
\newblock {\em Class. Quantum Grav.}, 12(9):2249--2274, 1995.

\bibitem[Ist81]{Istratescu81}
V.I. Istr\u{a}\c{t}escu.
\newblock {\em Fixed point theory. An introduction}.
\newblock D. Reidel Publishing Company, 1981.

\bibitem[Lea14]{Leach14}
J.~Leach.
\newblock A far-from-{CMC} existence result for the constraint equations on
  manifolds with ends of cylindrical type.
\newblock {\em Class. Quantum Grav.}, (31):035003, 2014.

\bibitem[Max09]{Maxwell09}
D.~Maxwell.
\newblock A class of solutions of the vacuum {E}instein constraint equations
  with freely specified mean curvature.
\newblock {\em Math. Res. Lett.}, 16(4):627--645, 2009.

\bibitem[Maz91]{Maz91}
R.~Mazzeo.
\newblock Elliptic theory of differential edge operators. i.
\newblock {\em Commun. Partial Diff. Eq.}, 16:1615--1664, 1991.

\end{thebibliography}
\end{document}